\documentclass[10pt,a4paper]{amsart}
\usepackage[margin=20mm]{geometry}
\usepackage[numbers]{natbib}
\usepackage{hyperref}
\usepackage{graphicx}
\usepackage{amssymb,amsmath}
\usepackage{amsthm}
\usepackage{bbm}
\usepackage{stmaryrd}
\usepackage[usenames,dvipsnames]{xcolor}
\usepackage{color}

\input xy
\xyoption{all} 

\numberwithin{equation}{section}

\newtheorem{theorem}{Theorem}[section]
\newtheorem{corollary}[theorem]{Corollary}

\newtheorem{proposition}[theorem]{Proposition}

\theoremstyle{definition}
\newtheorem{definition}[theorem]{Definition}
\newtheorem{remark}[theorem]{Remark}
\newtheorem{example}[theorem]{Example}

\newcommand{\Dleft}{[\hspace{-1.5pt}[}
\newcommand{\Dright}{]\hspace{-1.5pt}]}
\newcommand{\SN}[1]{\Dleft #1 \Dright}

\newcommand{\Id}{\mathbbmss{1}}
\newcommand{\p}{\mbox{\boldmath$\rho$}}

\newcommand{\rmd}{\textnormal{d}}

\DeclareMathOperator{\Vect}{Vect}

\DeclareMathOperator{\Der}{Der}

\font\black=cmbx10 \font\sblack=cmbx7 \font\ssblack=cmbx5 \font\blackital=cmmib10  \skewchar\blackital='177
\font\sblackital=cmmib7 \skewchar\sblackital='177 \font\ssblackital=cmmib5 \skewchar\ssblackital='177
\font\sanss=cmss10 \font\ssanss=cmss8 
\font\sssanss=cmss8 scaled 600 \font\blackboard=msbm10 \font\sblackboard=msbm7 \font\ssblackboard=msbm5
\font\caligr=eusm10 \font\scaligr=eusm7 \font\sscaligr=eusm5  \font\fraktur=eufm10
\font\sfraktur=eufm7 \font\ssfraktur=eufm5 
\font\bsymb=cmsy10 scaled\magstep2
\def\all#1{\setbox0=\hbox{\lower1.5pt\hbox{\bsymb
       \char"38}}\setbox1=\hbox{$_{#1}$} \box0\lower2pt\box1\;}
\def\exi#1{\setbox0=\hbox{\lower1.5pt\hbox{\bsymb \char"39}}
       \setbox1=\hbox{$_{#1}$} \box0\lower2pt\box1\;}

\def\tx#1{{\fam0\relax#1}}

\newfam\bifam
\textfont\bifam=\blackital \scriptfont\bifam=\sblackital \scriptscriptfont\bifam=\ssblackital

\newfam\blfam
\textfont\blfam=\black \scriptfont\blfam=\sblack \scriptscriptfont\blfam=\ssblack

\newfam\bbfam
\textfont\bbfam=\blackboard \scriptfont\bbfam=\sblackboard \scriptscriptfont\bbfam=\ssblackboard

\newfam\ssfam
\textfont\ssfam=\sanss \scriptfont\ssfam=\ssanss \scriptscriptfont\ssfam=\sssanss
\def\sss#1{{\fam\ssfam\relax#1}}

\newfam\clfam
\textfont\clfam=\caligr \scriptfont\clfam=\scaligr \scriptscriptfont\clfam=\sscaligr

\newfam\frfam
\textfont\frfam=\fraktur \scriptfont\frfam=\sfraktur \scriptscriptfont\frfam=\ssfraktur

\def\hpb#1{\setbox0=\hbox{${#1}$}
    \copy0 \kern-\wd0 \kern.2pt \box0}
\def\vpb#1{\setbox0=\hbox{${#1}$}
    \copy0 \kern-\wd0 \raise.08pt \box0}

\def\pmb#1{\setbox0\hbox{${#1}$} \copy0 \kern-\wd0 \kern.2pt \box0}
\def\pmbb#1{\setbox0\hbox{${#1}$} \copy0 \kern-\wd0
      \kern.2pt \copy0 \kern-\wd0 \kern.2pt \box0}
\def\pmbbb#1{\setbox0\hbox{${#1}$} \copy0 \kern-\wd0
      \kern.2pt \copy0 \kern-\wd0 \kern.2pt
    \copy0 \kern-\wd0 \kern.2pt \box0}
\def\pmxb#1{\setbox0\hbox{${#1}$} \copy0 \kern-\wd0
      \kern.2pt \copy0 \kern-\wd0 \kern.2pt
      \copy0 \kern-\wd0 \kern.2pt \copy0 \kern-\wd0 \kern.2pt \box0}
\def\pmxbb#1{\setbox0\hbox{${#1}$} \copy0 \kern-\wd0 \kern.2pt
      \copy0 \kern-\wd0 \kern.2pt
      \copy0 \kern-\wd0 \kern.2pt \copy0 \kern-\wd0 \kern.2pt
      \copy0 \kern-\wd0 \kern.2pt \box0}

\mathchardef\za="710B  
\mathchardef\zb="710C  
\mathchardef\zg="710D  
\mathchardef\zd="710E  
\mathchardef\zve="710F 
\mathchardef\zz="7110  
\mathchardef\zh="7111  
\mathchardef\zvy="7112 
\mathchardef\zi="7113  
\mathchardef\zk="7114  
\mathchardef\zl="7115  
\mathchardef\zm="7116  
\mathchardef\zn="7117  
\mathchardef\zx="7118  
\mathchardef\zp="7119  
\mathchardef\zr="711A  
\mathchardef\zs="711B  
\mathchardef\zt="711C  
\mathchardef\zu="711D  
\mathchardef\zvf="711E 
\mathchardef\zq="711F  
\mathchardef\zc="7120  
\mathchardef\zw="7121  
\mathchardef\ze="7122  
\mathchardef\zy="7123  
\mathchardef\zf="7124  
\mathchardef\zvr="7125 
\mathchardef\zvs="7126 
\mathchardef\zf="7127  
\mathchardef\zG="7000  
\mathchardef\zD="7001  
\mathchardef\zY="7002  
\mathchardef\zL="7003  
\mathchardef\zX="7004  
\mathchardef\zP="7005  
\mathchardef\zS="7006  
\mathchardef\zU="7007  
\mathchardef\zF="7008  
\mathchardef\zW="700A  
\mathchardef\zC="7009  

\newcommand{\be}{\begin{equation}}
\newcommand{\ee}{\end{equation}}

\newcommand{\bea}{\begin{eqnarray}}
\newcommand{\eea}{\end{eqnarray}}
\def\*{{\textstyle *}}
\newcommand{\R}{{\mathbb R}}

\newcommand{\Z}{{\mathbb Z}}

\newcommand{\s}{{\textstyle *}}

\def\Sec{\sss{Sec}}
\def\Vect{\sss{Vect}}

\def\sT{{\sss T}}

\def\xi{\tx{i}}

\def\s*{{\scriptstyle *}}

\def\cO{\mathcal{O}}

\newcommand{\beas}{\begin{eqnarray*}}
\newcommand{\eeas}{\end{eqnarray*}}

\def\half{\frac{1}{2}}

\title{Homological sections of Lie algebroids}

   \author{Andrew James Bruce} 
   \address{Current affiliation: Department of Mathematics, The Computational Foundry,
Swansea University Bay Campus, Fabian Way, Swansea, SA1 8EN, United Kingdom \newline \indent Research conducted at the  Department of Mathematics, University of Luxembourg, Maison du Nombre 6, avenue de la Fonte, 
L-4364 Esch-sur-Alzette, Luxembourg}  
   \email{andrewjamesbruce@googlemail.com}
  
\date{\today}
 
\begin{document}

\begin{abstract}
We examine  Lie (super)algebroids equipped with a \emph{homological section}, i.e., an odd section that `self-commutes', we refer to such Lie algebroids as \emph{inner Q-algebroids}: these provide natural examples of suitably ``superised'' Q-algebroids in the sense of Mehta. Such Lie algebroids are a natural generalisation of Q-manifolds and Lie superalgebras equipped with a homological element. Amongst other results,  we show that, via the derived bracket formalism, the space of  sections of an inner Q-algebroid comes equipped with an odd Loday--Leibniz bracket.   \par
\end{abstract}

 \maketitle

\setcounter{tocdepth}{2}
\vspace{-30pt}

\section{Introduction} 
Q-manifolds  are supermanifolds equipped with an odd vector field, usually denoted as $Q$ and referred to as a homological vector field, that satisfies $2 Q^2 = [Q,Q] =0$ (see \cite{Alexandrov:1997,Schwarz:1993,Shander:1980,Vaintrob:1996}). As we have a $\Z_2$-graded Lie bracket on the module of vector fields this condition is non-trivial.  Such supermanifolds appear in physics via the BRST formalism  \cite{Becchi:1976,Tyutin:1975}, the homological vector field is the BRST differential. In mathematics, Q-manifolds offer a neat  formulation of  Lie algebroids and some of their relatives such as Courant algebroids (see \cite{Roytenberg:2002,Vaintrob:1997}). Formulating classical notions in terms of Q-manifolds often allows for a clear geometric picture to emerge, for example a simple notion of Lie algebroid morphisms is obtained via this approach.  Another example of the power of employing Q-manifolds is Kontsevich's  proof of the formality theorem, which is essential in his proof that all Poisson manifolds admit a deformation quantisation \cite{Kontsevich:2003}. There are many other instances where the `Q-philosophy' has shed light on difficult questions, a \emph{very} incomplete list includes \cite{Bonavolonta:2013, Dotsenko:2016, Khudaverdyan:2014,Voronov:2002, Voronov:2012}.  The classical understanding of a Lie algebroid, due to Pradines \cite{Pradines:1974}, is as a vector bundle together with a Lie bracket on its space of sections and an anchor map from the vector bundle to the tangent bundle of the base manifold, that satisfy some natural compatibility conditions (see Definition \ref{def:LieAlg}). A Lie algebroid can be seen as a generalisation of a tangent bundle and its sections as a generalisation of vector fields. The general mantra here is that the tangent bundle of a manifold can be replaced with a Lie algebroid in all the standard constructions of differential geometry. Classical Lie algebroids generalise verbatim to the setting of supermanifolds; in practice, this is just taking care of the extra sign factors that arise. Thus, we have the ``superalgebroid creed'' -  \emph{notions and constructions that involve the tangent bundle of a supermanifold generalise to a Lie superalgebroid}. In particular, one should be able to replace vector fields on a supermanifold with sections of a Lie superalgebroid.  \par 
In this paper, we give the notion of a \emph{homological section}  of a Lie superalgebroid (see Definition \ref{def:HomSec}), i.e., an odd section that ``self-commutes'' with respect to the ($\Z_2$-graded) Lie bracket, and examine some direct consequences thereof. This condition is, in general, just as the case for homological vector fields, non-trivial. Lie superalgebroids equipped with such a section, \emph{{ inner Q-algebroids}}\footnote{ we refrain from using the term ``dg Lie algebroid'' as we will not consider objects with an additional $\Z$ or $\mathbb{N}$ grading.} as we refer to them, are direct generalisations of Q-manifolds. They should also simultaneously be viewed as the oidification (or horizontal categorification) of Lie superalgebras equipped with a homological element (see \cite{Lebedev:2005} and Example \ref{exp:LieAlgQ}).  In essence, { inner Q-algebroids} are the `many object' version of inner differential Lie superalgebras.  { Inner Q-algebroids are natural examples of ``superised''  Q-algebroids in the sense of Mehta (see \cite{Mehta:2006,Mehta:2009}). Essentially, if the de Rham complex of a Lie algebroid comes with an additional compatible differential then we have a Q-algebroid. For the case at hand, the differential is an inner differential generated by the homological section and the Lie superbracket. A large motivating factor for the study of Q-algebroids is Mackenzie's theory of double Lie algebroids and related structures coming from Poisson geometry. Indeed, we will show a tight relation between inner Q-algebroids and triangular Lie bialgebroids. Moreover, given how homological vector fields have proven extremely versatile in algebra and geometry we conjecture that homological sections may similarly be useful.}  \par 
We remind the reader that sections of an integrable Lie algebroid correspond to left-invariant vector fields (i.e., vector fields tangent to the source fibres) on the integrating Lie groupoid. Thus, assuming the Lie superalgebroid is integrable, then a homological section corresponds to a left-invariant homological vector field on the integrating Lie supergroupoid. An obvious example here would be the pair supergroupoid of a Q-manifold.  It is fair to say that Lie supergroupoids and compatible geometrics structures on them are not well-studied objects. Mehta, to our knowledge, was the first to consider multiplicative homological vector fields on graded groupoids (see \cite{Mehta:2009b}). In particular, he established the correspondence between Mackenzie's $\mathcal{LA}$-groupoids (see \cite{Mackenzie:1992}) and, what are known as Q-groupoids.  We will not consider Lie supergroupoids  past this point in this paper. \par 
It is well known that many deformation problems are controlled by homological elements of a Lie superalgebra, or rather the cohomology they define. This observation goes back to at least De Wilde \& Lecomte \cite{DeWilde:1988} in the context of deformation quantisation. In particular, they looked for graded Lie algebras such that there is a one-to-one correspondence between the algebraic structures they were interested in and degree one homological elements of the graded Lie algebra. The idea that algebraic notions can be encoded in this way has been generalised in many directions, including homotopy algebras over any quadratic Koszul operad (see Ginzburg \& Kapranov \cite{Ginzburg:1994}). A particular  realisation for this is Voronov's construction of $L_\infty$-algebras via higher derived brackets (see \cite{Voronov:2005}).  Homotopy Leibniz algebras were also given in similar terms by Ammar \& Poncin (see \cite{Ammar:2010}). \par 
Another motivation for this work comes from Poisson/Schouten (super)geometry. In particular, in the classical setting, the canonical Poisson structure on $\sT^* M$, the canonical Schouten structure on $\Pi \sT^* M$, and the de Rham differential (a homological vector field) on  $\Pi \sT M$ are equivalent manifestations of the smooth structure on $M$. Poisson and Schouten structures on Lie superalgebroids, including their higher versions have been studied (for the higher case see \cite{Bruce:2010}). These structures can also themselves be encoded in particular homological vector fields. In order to complete this picture it is natural to consider homological sections of Lie superalgebroids. \par 
From this point on we will generally drop the prefix `super' and understand, unless otherwise explicitly stated, that all objects are $\Z_2$-graded. For example, by a manifold we will generally mean a supermanifold, and by a Lie algebroid we will generally mean a Lie superalgebroid. In this paper we establish, amongst other results, the following:
\begin{enumerate}
\item  As a direct result of the definition of a homological section, we observe that we have a natural bi-complex on Lie algebroid forms, and accordingly, { inner Q-algebroids}  give a natural class of, suitably ``superised'' Q-algebroids (see Proposition \ref{prop:IQAgebroid}).
 It should be remarked that other natural class of Q-algebroids are Mackenzie's double Lie algebroids (see \cite{Mackenzie:1992} and Voronov \cite{Voronov:2012}).  Thus, new examples of Q-algebroids can be constructed as { inner Q-algebroids}.
 \item { We show that from a homological section one can build a Poisson structure on the Lie algebroid. Thus, we have a triangular Lie bialgebroid in the sense of Mackenzie and Xu (see \cite{Mackenzie:1994}), see Theorem \ref{thm:BiAlg}. }
 \item  A homological section of a Lie algebroid clearly induces a differential on the Lie superalgebra of sections, see Proposition \ref{prop:DiffLieAlg}. In turn, this implies, via the derived bracket formalism (see \cite{KosmannSchwarzbach:1996,KosmannSchwarzbach:2004}), that the sections of an { inner Q-algebroid} come equipped with an odd Loday--Leibniz bracket, see Proposition \ref{trm:OddLodLei}. It should be noted that we do not obtain a new  Lie algebroid structure (upto a shift in parity) as we do not have the required symmetry. Moreover, the right and left Leibniz rules are different. 
\item If the modular class of the Lie algebroid vanishes, then the modular class associated with the Lie derivative of a homological section also vanishes, see Proposition \ref{prop:ModClass}. This implies that one can find a Berezin volume that is both invariant under the associated de Rham differential and the Lie derivative with respect to the homological section, see Corollary \ref{cor:InvBerVol}.
\end{enumerate}
\noindent \textbf{Preliminaries.} We assume that the reader is familiar with the basic theory of supermanifolds (see \cite{Berezin:1975}), i.e., a supermanifold is a locally super-ringed space that is locally isomorphic to $\mathbb{R}^{n|m} := \big (\R^{n}, C^{\infty}_{\R^n}(-)\otimes \Lambda(\theta^{1}, \cdots \theta^{m}) \big)$. Here, $\Lambda(\theta^{1}, \cdots \theta^{m})$ is the Grassmann algebra (over $\R$) with $m$ generators and $C^{\infty}_{\R^n}(-)$ is the sheaf of smooth function on $\R^n$. We will denote a supermanifold by $M = (|M|, \cO_M)$, where $|M|$ is the underlying reduced manifold and $\cO_M$ is a sheaf of superalgebras. We denote the Grassmann parity of objects by `tilde', i.e., $\widetilde{O} \in \Z_2$. We will denote the algebra of global sections of a supermanifold as $C^\infty(M) :=  \cO_M(|M|)$, and refer to then as functions. Morphisms of supermanifolds are morphisms of  super-ringed spaces. That is, a morphism $\phi : M \rightarrow N$ consists of a pair $ \phi = (|\phi|, \phi^*  )$, where $|\phi| : |M| \rightarrow |N|$ is a continuous map (in fact, smooth) and  {$\phi^*: \cO_N \rightarrow |\phi|_* \cO_M$  is a morphism of sheaves of superalgebras}. Given any point on $|M|$ we can always find a `small enough' open neighbourhood $|U|\subseteq |M|$ such that we can  employ local coordinates $x^{a} := (x^{\mu} , \theta^{i})$. One can work with local coordinates in more-or-less in the same way as one does on manifolds when it comes to describing geometric objects and morphisms. The $C^\infty(M)$-module of vector fields we denote as $\Vect(M)$. \par 
{ Recall that a vector bundle in the category of supermanifolds $\tau : E \rightarrow M$ is a supermanifold equipped with an adapted atlas with charts admitting coordinates of the form  $(x^a, y^\alpha)$, the parity of these are denoted  $\widetilde{x^a} = \widetilde{a}$ and $\widetilde{y^\alpha} = \widetilde{\alpha}$.  In particular, a vector bundle is locally isomorphic the product $U \times \R^{r|s}$, where $U:= (|U|, \cO_M|_{|U|})$ and $|U| \subset |M|$ is open. The admissible changes of coordinates are of the form (using the standard abuses of notation)
\begin{align}\label{eqn:CoordChan}
& x^{a'} = x^{a'}(x), && y^{\alpha'} = y^\beta T_{\beta}^{\,\, \alpha'}(x)\,
\end{align}
where $T_\beta^{\,\, \alpha'}$ are even invertible matrices (with entries being local sections of $M$). 
}
 The space of sections, both even and odd, of a (geometric) vector bundle $\tau: E \rightarrow M$ we will denote as $\Sec(E)$. The even sections can be understood as `geometric sections' of $E$, while odd sections as `geometric sections' of $\Pi E$. { Here $\Pi E$ is the parity reversed bundle, i.e., the parity of the fibre coordinates is shifted with respect to the original vector bundle.  Thus, even sections are morphisms of supermanifolds $s :  M \rightarrow E$ such that $s\circ \tau = \Id_M$. Similarly, odd sections are morphisms of supermanifolds $s : M \rightarrow E$ such that $s\circ \tau^\Pi = \Id_M$.} We remark that we have a supermanifold version of the smooth Serre-Swan theorem (for a sketch of the proof see \cite{Balduzzi:2011}), which essentially states the equivalence of geometric vector bundles and algebraic vector bundles (i.e., sheaves of locally free $\cO_M$-modules).  \par
We also take it for granted that the reader has a grasp of the theory of Lie algebroids including their numerous examples. Our standard reference for Lie groupoids and Lie algebroids is the book by Mackenzie \cite{Mackenzie:2005}. For completeness, we will give the definition of a Lie algebroid (adapted to the setting of supermanifolds). 
\begin{definition}[Pradines \cite{Pradines:1974}]\label{def:LieAlg}
A \emph{Lie algebroid} is a triple $(A, [-,-], \rho)$, where $\pi : A \rightarrow M$ is a vector bundle, $[-,-]$ is a Lie bracket on the vector space $\Sec(A)= \Sec_0(A)\oplus \Sec_1(A)$ { that preserves the Grassmann parity}, referred to as the \emph{Lie algebroid bracket}, and  a linear map $\rho: \Sec(A) \rightarrow \Vect(M)$, referred to as the \emph{anchor}, that satisfies the Leibniz rule
$$[u,f\,v] = \rho(u)f \, v  + (-1)^{\widetilde{u} \widetilde{f}} \, f \, [u,v],$$
with $u$ and $v \in \Sec(A)$ and $f \in C^\infty(M)$. 
\end{definition}
It should be noted that one can deduce that the anchor is in fact a homomorphism of Lie algebras, i.e., $\rho[u,v] = [\rho(u), \rho(v)]$. { In a more algebraic language, if a vector bundle $A$ is a Lie algebroid then $\big(C^\infty(M), \Sec(A) \big)$ is a Lie--Rinehart algebra (see \cite{Rinehart:1963}).} By picking a `small enough' open submanifold $U = (|U|, \cO_M|_{|U|})$ such that one can  employ local coordinates $x^a$ together with a local basis $\{ t_\alpha\}$ of $\Sec(A)$ and defining
\begin{align*}
[t_\alpha, t_\beta] = (-1)^{\widetilde{\beta}} \, Q_{\alpha \beta}^\gamma(x)t_\gamma, && \rho(t_\alpha) = Q_\alpha^a(x)\frac{\partial}{\partial x^a},
\end{align*}
the Lie algebroid bracket { for arbitrary sections $u = u^\alpha(x)t_\alpha$ and $v = v^\alpha(x)t_\alpha$} can (locally) be written as
\begin{equation}\label{eqn:LocForBra}
[u,v] = \left( u^\alpha(x)Q^a_\alpha(x)\frac{\partial v^\gamma(x)}{\partial x^a} - (-1)^{\widetilde{u} \widetilde{v}} \, v^\alpha(x)Q^a_\alpha(x)\frac{\partial u^\gamma(x)}{\partial x^a}  - (-1)^{\widetilde{\alpha}(\widetilde{v}+1)} \, u^\alpha(x) v^\beta(x) Q^\gamma_{\beta \alpha}(x)\right )t_\gamma\,.
\end{equation}
For brevity, we will often just write $A$ for a Lie algebroid when no confusion can arise. It is well-known that Lie algebroids have a very economic description in terms of Q-manifolds \`{a} la Vaintrob \cite{Vaintrob:1997}.
The { parity reversed bundle} $\Pi A$ can then be defined as the supermanifold equipped with local coordinates $(x^a, \zx^\alpha)$, where the parity of the fibre coordinates has now been shifted, i.e., $\widetilde{\zx^\alpha} = \widetilde{\alpha} +1$. The admissible changes of coordinates are the same as \eqref{eqn:CoordChan} upon replacing $y$ with $\zx$. As the  changes of the fibre coordinates on $A$ are linear, this replacement is valid and the resulting supermanifold $\Pi A$ is well-defined. The parity reversed bundle  {-  just as any vector bundle does -} comes with a natural $\mathbb{N}$-grading, which we refer to as \emph{weight}, defined by assigning weight zero to the coordinates on the base manifold and weight one to the fibre coordinates. This assignment of weight is independent of the Grassmann parity of the coordinates (see \cite{Voronov:2002} for further details).  A Lie algebroid structure on a vector bundle $A$ is equivalent to a weight one homological vector field on $\Pi A$.  By minor abuse of language, we will also refer to $(\Pi A, Q)$ as a Lie algebroid. The bracket and anchor can be recovered using the derived bracket formalism  of Kosmann-Schwarzbach \cite{KosmannSchwarzbach:1996,KosmannSchwarzbach:2004} (also see Voronov \cite{Voronov:2005}). \par 
The general set-up of the derived bracket formalism (slightly adapted to our current needs) is as follows. Given any differential Lie superalgebra $(\mathfrak{g}, [-,-], \delta)$ we have a \emph{derived bracket} on $\mathfrak{g}$  given by  $(a,b) :=  (-1)^{\widetilde{a}} [\delta a, b]$. This bracket is Grassman odd and in general \emph{not} graded skewsymmetric. It does however satisfy an appropriate version of the Jacobi identity, and thus we have an odd Loday--Leibniz bracket on $\mathfrak{g}$. In practice, the differential is often an inner derivation. That is, $\delta = [q, -]$ for some odd element $q$. In order for this to be a differential, i.e., $\delta^2 =0$, it is clear that we must have $[q,q]=0$. In other words, $q$ must be homological. \par 
We { adapt} the definition of a Q-algebroid as given by Mehta \cite{Mehta:2006,Mehta:2009} to the setting of supermanifolds (his definitions were originally given in the category of $\Z$-graded manifolds). Essentially we view a Q-algebroid as a Lie algebroid equipped with a compatible homological vector field.  In a more categorical language, a Q-algebroid is a Lie algebroid in the category of Q-manifolds, and vice versa. We remark that the existence of an  { inner Q-algebroid} structure implies the existence of a Q-algebroid structure (see Proposition \ref{prop:IQAgebroid}). 
\begin{definition}[Mehta \cite{Mehta:2006,Mehta:2009}]\label{def:QAlgb}
A Lie algebroid $(\Pi A,Q)$ equipped with a homological vector field of weight zero, $\Xi \in \Vect(\Pi A)$ that commutes with $Q$, i.e.,  $[Q, \Xi]=0$, is referred to as a \emph{Q-algebroid}.
\end{definition}
By definition, a homological vector field is odd and so the condition that $[\Xi, \Xi]=0$ is non-trivial. Furthermore, note that the weight zero condition ensures that $\Xi$ is projectable and that its projection is also homological.  In local coordinates $(x^a, \zx^\alpha)$ on $\Pi A$ the two homological vector fields are of the form 
\begin{align*}
& Q = \zx^\alpha Q^a_\alpha(x)\frac{\partial}{\partial x^a} + \frac{1}{2}\zx^\alpha \zx^\beta Q^\gamma_{\beta \alpha}(x)\frac{\partial}{\partial \zx^\gamma}\,,
&& \Xi = \Xi^a(x) \frac{\partial}{\partial x^a} + \zx^\alpha \Xi_\alpha^{\,\,\beta}(x) \frac{\partial}{\partial \zx^\beta}\,.
\end{align*}

\section{Homological sections and  inner Q-algebroids}
\subsection{ Inner Q-algebroids: definition and examples}
We now proceed to give the main definition of this paper and explore some direct consequences thereof. 
\begin{definition}\label{def:HomSec}
Let $(A, [-,-], \rho)$ be a Lie algebroid. A \emph{homological section} of $A$ is an odd section $q \in \Sec(A)$ that  {`self-commutes', i.e., $[q,q]=0$}. An \emph{{ inner Q-algebroid}} is a Lie algebroid equipped with a distinguished  homological section. 
\end{definition}
We will denote an { inner Q-algebroid} as $(A,q)$ where the Lie algebroid bracket and anchor map are implied. For brevity, we will write { IQA} for an { inner Q-algebroid}. Just as in the case of homological vector fields, the condition $[q,q] =0$ is, in general, not automatic. 
\begin{proposition}\label{prop:DiffLieAlg}
Let $(A,q)$ be an { IQA} and set $\delta := [q,-]$. Then { $\delta^2 =0$, and $\delta$ is a derivation over the Lie algebroid bracket.}
\end{proposition}
\begin{proof}
It is clear that as we have a homological section  $\delta^2=0$. Thus,  we need only check that $\delta$ acts as a derivation over the Lie algebroid bracket. However, this is obvious in light of the Jacobi identity, that is, we have an inner derivation. Explicitly,
$$\delta [u,v] = [q,[u,v]] = [[q,u],v] +(-1)^{\widetilde{u}}\, [u, [q,v]]
 = [\delta u, v] +(-1)^{\widetilde{u}} \, [u, \delta v].$$
\end{proof}
\begin{proposition}\label{prop:Qman}
Let $(A,q)$ be an { IQA} over $M$, i.e., $\pi: A \rightarrow M$, and set $Q_q := \rho(q) \in \Vect(M)$. Then $(M, Q_q)$ is a Q-manifold.
\end{proposition}
\begin{proof}
This follows directly from the fact that { $q$ is homological section} and that the  anchor map is a homomorphism of Lie algebras, specifically
$$[\rho(q), \rho(q)] = \rho([q,q]) =0.$$
\end{proof}
Using \eqref{eqn:LocForBra}, locally the condition for a section $q = q^\alpha(x)t_\alpha$ to be homological can be written as
$$\left(q^\alpha Q_\alpha^a \frac{\partial q^\gamma}{\partial x^a}-\frac{1}{2}q^\alpha q^\beta Q^\gamma_{\beta \alpha}\right) t_\gamma=0.$$
Similarly, locally $Q_q = q^\alpha Q_\alpha^a \frac{\partial}{\partial x^a},$
and the fact that it `squares to zero' is expressed as
$$Q_\alpha^a\left( \frac{\partial q^\beta}{\partial x^a}Q_\beta^b + (-1)^{\widetilde{a}(\widetilde{\beta}+1)}\, q^\beta \frac{\partial Q_\beta^b}{\partial x^a} \right) =0.$$
\begin{remark}
 The homological vector field $Q_q \in \Vect(M)$ may be the zero vector field. This will be the case if  $M$ is a pure even manifold. One can easily see that the condition for a section to be homological becomes purely algebraic if the base is a pure even manifold.  Specifically, if we pick a local basis $\{t_l, \tau_\mu \}$ with $\widetilde{t} = 0$ and $\widetilde{\tau} =1$, then an odd section is locally of the form $q = q^\mu(x)\tau_\mu$. Then the homological condition is $q^\mu(x)q^\nu(x)Q^l_{\nu \mu}(x)=0$. Note that as $Q^l_{\nu \mu} = Q^l_{\mu\nu}$, in general, this condition is non-trivial.
 \end{remark}
\begin{example}[Zero Lie algebroids]
Consider a vector bundle $\pi : A \rightarrow M$ in the category of supermanifolds. If we equip this with the zero bracket and zero anchor then we have what we will refer to as a zero Lie algebroid. Any odd section of $A$ is by definition a homological section.
\end{example}
\begin{example}[Trivial homological sections]
Let $(A, [-,-], \rho)$ be a Lie algebroid. Then the zero section of $A$ is  homological. We refer to this as the trivial homological section. 
\end{example}
\begin{example}[Q-manifolds]
Any Q-manifold $(M,Q)$ can be considered as the { IQA} $\big(\sT M, [-,-], \Id_{\Vect(M)} , Q\big)$.
\end{example}
\begin{example}[Lie algebras with a nilpotent element]\label{exp:LieAlgQ}
A (finite-dimensional) Lie algebra $(\mathfrak{g}, [-,-]) $ can be considered as a Lie algebroid over a point.  By replacing the $\Z_2$-graded $\R$-vector space $\mathfrak{g}$ with the corresponding linear supermanifold, which we denote as $^ \textrm{Sm}\mathfrak{g}$, the vector bundle structure is 
$$\pi :~ ^ \textrm{Sm}\mathfrak{g} \rightarrow \{ \star\}.$$ 
We can identify $\Sec(^ \textrm{Sm}\mathfrak{g})$ with the original vector space $\mathfrak{g}$.  Thus, a nilpotent element is an odd section of $^\textrm{Sm}\mathfrak{g}$, or equivalently, an odd element of $\mathfrak{g}$ such that $[q,q]=0$. As a specific example, consider $\mathfrak{gl}(1|1)$ which has two even { basis} elements $N$ and $Z$, and two odd  { basis} elements $\psi^{\pm}$, with the non-zero Lie brackets being $[N, \psi^{\pm}] = \pm 2 \, \psi^\pm$ and $[\psi^+, \psi^-] = Z$. In particular, both $\psi^\pm$ are nilpotent, i.e., $[\psi^\pm, \psi^\pm]=0$.
\end{example}
\begin{remark}
Using the super-version of Lie III, Example \ref{exp:LieAlgQ}, corresponds to a Lie supergroup equipped with a left-invariant homological vector field. 
\end{remark}
The above examples are the four extremes. Note that non-trivial homological sections cannot exist in the category of `ordinary' non-super Lie algebroids.
\begin{example}[SUSY action Lie algebroid]\label{exa:SUSYActionAlg}
Consider the $N=1$ $d =1$ supertranslation algebra $\mathfrak{g}_{\textrm{S}}$ which, in terms of the  even and odd generators $t$ and $\tau$, we write as  $[t, \tau] =0$, $[\tau, \tau ] = t$, with all other brackets being zero. We define an action  $\rho : \mathfrak{g}_{\textrm{S}} \rightarrow \Vect(\R^{1|1})$ as 
\begin{align*}
& \rho(t) = \frac{\partial }{\partial x}, &&  \rho(\tau) = \frac{\partial}{\partial \theta} + \frac{1}{2} \theta \frac{\partial}{\partial x},
\end{align*}
where we have employed global coordinates $(x, \theta)$ on $\R^{1|1}$. The vector bundle structure is the trivial one $^{\textrm{Sm}}\mathfrak{g}_{\textrm{S}} \times \R^{1|1}\rightarrow \R^{1|1}$.  Sections of this vector bundle can be considered as maps from $\R^{1|1}$ to $\mathfrak{g}_{\textrm{S}}$, i.e., they are objects of the form $u = u_t(x, \theta)t + u_\tau(x, \theta)\tau$. The Lie algebroid bracket is constructed in the same way as the classical case. Specifically, 
$$[u,v] =(-1)^{\widetilde{v}+1}u_\tau v_\tau \, t + \big(\rho(u)v_t - (-1)^{\widetilde{u}\widetilde{v}} \, \rho(v)u_t \big)\, t + \big(\rho(u)v_\tau - (-1)^{\widetilde{u}\widetilde{v}} \, \rho(v)u_\tau \big)\, \tau\,.$$
Any odd section of the form $q = \theta q(x)\, t$ is homological. To see this observe that $[q,q]  = 2 \theta^2 q(x) \, q'(x)\, t =0$. The associated homological vector field is $Q_q = \theta \, q(x) \frac{\partial}{\partial x}$.
\end{example}
\begin{example}[Cotangent bundles]
Consider a manifold $M$ equipped with an even Riemannian or even symplectic structure. Then associated with these structures are the musical isomorphisms
\begin{align*}
& \flat : \Vect(M) \rightarrow \Omega^1(M), && \sharp : \Omega^1(M) \rightarrow \Vect(M), 
\end{align*}
where we are considering ``even'' one-forms, that is, the local basis $\delta x^a$ is taken to be of the same parity as the corresponding $\partial_a$. Then $\pi: \sT^*M \rightarrow M$ is a Lie algebroid where we define $\rho = \sharp$, and $[\alpha, \beta] = \flat \, [\, \sharp \alpha, \sharp \beta \, ]$.  A homological one-form corresponds to a homological vector field, i.e., $q = \flat Q$.
\end{example}
\begin{small}
\noindent \textbf{Aside.} The cousin of  `homological' is `supersymmetric'. Loosely, supersymmetry is an odd element of a Lie algebra that is the ``square root'' of some chosen even element. In particular, we can consider a Lie algebroid  equipped with an even and an odd section, $h$ and $q$, respectively, that satisfy  $[q,q] = h$. The vector fields on the base manifold also carry a representation of supersymmetry, i.e., $\rho(h) =:\mathcal{H}$ and $\rho(q) =: \mathcal{Q}$ satisfy the $d=1$ $N=1$ supertranslation algebra, that is, $[\mathcal{Q}, \mathcal{Q}] = \mathcal{H}$.  This is a generalisation of  the SUSY action Lie algebroid (see Example \ref{exa:SUSYActionAlg}).
\end{small}
\subsection{An associated bi-complex and Q-algebroids}\label{subsec:BiComp}
We will change our point of view slightly and consider $(\Pi A,Q)$ to be the `correct' starting definition of a Lie algebroid (see \cite{Vaintrob:1997}). Recall that we have an odd isomorphism
\begin{align*}
&\iota : ~  \Sec(A) \rightarrow \Vect_{-}(\Pi A)\\
 & u = u^\alpha(x)t_\alpha \mapsto \iota_u = (-1)^{\widetilde{u}}\, u^\alpha(x)\frac{\partial}{\partial \zx^\alpha},
\end{align*}
where $ \Vect_{-}(\Pi A)$ denotes the space of weight minus one vector fields on $\Pi A$. Note that this forms an abelian Lie algebra under the standard Lie bracket of vector fields.  The  anchor map and Lie algebroid bracket are recovered using the derived bracket formalism (see \cite{KosmannSchwarzbach:1996,Voronov:2005})
\begin{align*}
& \rho(u)f = [Q, \iota_u]f = (-1)^{\widetilde{u}} \, \iota_u(Qf),
&& \iota_{[u,v]} = (-1)^{\widetilde{u}} \, [[Q,\iota_u], \iota_v]. 
\end{align*}
 The Lie derivative of a section is defined via Cartan's magic formula, i.e., $L_u :=  [Q,\iota_u]$. Furthermore, a quick calculation shows that 
\begin{align*}
& [Q,L_u] = \frac{1}{2}[[Q,Q], \iota_u] = 0,\\
&[L_u, L_v] =   (-1)^{\widetilde{u}} \, [Q, [ [Q,\iota_u], \iota_v]] - (-1)^{\widetilde{u}}\, [[Q,L_u], \iota_v] = L_{[u,v]}.
\end{align*}
From the above, we observe that, for an { IQA} $(A,q)$,
$$Q_q(f) = [Q,\iota_q]f = - \iota_q(Qf).$$
Moreover, it is clear that
\begin{align}\label{egn:ComRel}
[Q,L_q] = 0, &&\textnormal{and}&& [L_q, L_q] = 0.
\end{align}
 Locally we have
\begin{align} \label{eq:Q}
& Q = \zx^\alpha Q^a_\alpha(x)\frac{\partial}{\partial x^a} + \frac{1}{2}\zx^\alpha \zx^\beta Q^\gamma_{\beta \alpha}(x)\frac{\partial}{\partial \zx^\gamma}\,,\\ \label{eq:Lq}
& L_q = q^\alpha(x) Q^a_\alpha(x) \frac{\partial}{\partial x^a} - \zx^\alpha \left( Q^a_\alpha(x) \frac{\partial q^\gamma}{\partial x^a}(x) - q^\beta(x)Q^\gamma_{\beta \alpha}(x) \right)\frac{\partial}{\partial \zx^\gamma}\,.
\end{align}
\begin{proposition}\label{prop:IQAgebroid}
Let $(A,q)$ be an { IQA}. Then $(\Pi A, Q, L_q)$ is a Q-algebroid, see Definition \ref{def:QAlgb}.
\end{proposition}
\begin{proof}
Comparing with Definition \ref{def:QAlgb}, it is clear from \eqref{egn:ComRel} that we have a Q-algebroid.
\end{proof}
\begin{remark}\label{rem:LinftyALg}
Clearly, as the pair of homological vector fields commute, $\widehat{Q}:=  L_q + Q \in \Vect(\Pi A)$ is itself a homological vector field, albeit inhomogeneous in weight. Thus, once a homological section has been chosen, we have an $L_\infty$-algebroid structure on $A$ concentrated in degrees zero and one (see \cite{Bruce:2011,Khudaverdian:2008}). 
\end{remark}
\begin{example}
Given  a Q-manifold $(M,Q)$, from the Cartan calculus, in particular $[\rmd, L_X] =0$ and $[L_X,L_Y]=L_{[X,Y]}$ for abrbitary vector fields $X$ and $Y \in \Vect(M)$,  it is clear that $(\Pi TM, \rmd, L_Q)$, is a Q-algebroid.
\end{example}
 The space of Lie algebroid forms we  define as $\Omega^\bullet(A) :=\mathcal{A}(\Pi A) = \bigoplus \mathcal{A}^p(\Pi A) $, i.e., fibre-wise polynomials on $\Pi A$, or in other words,  functions of homogeneous weight. This is essentially no different to defining differential forms on a manifold $M$ in terms of function on $\Pi \sT M$. Lie algebroid forms naturally form a $(\Z_2, \mathbb{N})$-graded algebra where the commutation rules are defined in via the Grassmann parity.  We will denote Lie algebroid $p$-forms of parity $i$ as $\Omega^p_i(A)$.  Note that $Q$ is weight one, while $L_q$ is of weight zero, but both have parity one. It is clear that $(\Omega^\bullet(A), Q, L_q)$ is a bi-differential algebra, i.e., a (super)commutative algebra equipped with two differentials that mutually (anti)commute.    We  have a bi-complex (or double complex) on Lie algebroid forms { as first defined by Mehta for general Q-algebroids} (see \cite{Mehta:2009})
 \begin{equation}\label{eqn:bicomp}
\leavevmode
\begin{xy}
(0,20)*+{ \Omega^0_{0/1}(A)}="a"; (30,20)*+{\Omega^1_{1/0}(A)}="b"; (60,20)*+{ \Omega^2_{0/1}(A)}="c"; (90,20)*+{\cdots}="d";
(0,0)*+{\Omega^0_{1/0}(A)}="e"; (30,0)*+{\Omega^1_{0/1}(A)}="f";  (60,0)*+{\Omega^2_{1/0}(A)}="g"; (90,0)*+{\cdots}="h";  %
{\ar "a";"b"}?*!/_3mm/{Q};%
{\ar "b";"c"}?*!/_3mm/{Q};%
{\ar "c";"d"}?*!/_3mm/{Q};%
{\ar "e";"f"}?*!/^3mm/{Q};%
{\ar "f";"g"}?*!/^3mm/{Q};%
{\ar "g";"h"}?*!/^3mm/{Q};%
{\ar@<1.ex>"a";"e"}; {\ar@<1.ex> "e";"a"};?*!/_5mm/{L_q};%
{\ar@<1.ex>"b";"f"}; {\ar@<1.ex> "f";"b"};?*!/_5mm/{L_q};%
{\ar@<1.ex>"c";"g"}; {\ar@<1.ex> "g";"c"};?*!/_5mm/{L_q};%
\end{xy}
\end{equation}
Note that $\Omega^0(A) = C^\infty(M)$ and that $L_q f = Q_q(f)$ for any function $f \in C^\infty(M)$ and so the first vertical complex  is the standard complex of $(M, Q_q)$. Note  that there are no  Lie algebroid top forms (unless we have a pure even Lie algebroid) and so the complex is unbounded.
\begin{example}\label{exa:SUSYActionAlg2}
Continuing Example \ref{exa:SUSYActionAlg}, the SUSY action Lie algebroid, as a Q-manifold, is given by $\big(\Pi(^{\textrm{Sm}}\mathfrak{g}_S) , Q  \big)$, where in global coordinates $(x, \theta, \zx, z)$, where $\widetilde{x} = \widetilde{z} =0$ and  $\widetilde{\theta} = \widetilde{\zx} =1$, the homological vector field is given by
$$Q = \left(\zx + \half z \theta\right)\frac{\partial}{\partial x} + z \frac{\partial}{\partial \theta} - \frac{1}{2}z^2 \frac{\partial}{\partial \zx}.$$
A quick calculation shows that $Q^2=0$. We argued that any odd section of the form $q = \theta \,q(x) \,t$ is homological. Using the odd isomorphism we have
$$\iota_q = - \theta \, q(x) \frac{\partial}{\partial \zx}\,,$$
and
$$L_q = \theta \, q(x) \frac{\partial}{\partial x} - \left( \zx \theta \,\frac{\partial q(x)}{\partial x} + z \, q(x) \right)\frac{\partial}{\partial \zx}\,.$$ 
We remark that odd p-forms that do not contain $\zx$, i.e., $\omega = z^p \theta \, \omega(x)$, this is well-defined as we only consider linear changes of fibre coordinates, are  $L_q$-closed. 
\end{example}
{
The \emph{total supercomplex} of an inner Q-algebroid is the supercomplex
\begin{equation*}
\leavevmode
\begin{xy}(0,40)*+{\Omega^\bullet_0(A)}="a"; (30,40)*+{\Omega^\bullet_1(A)}="b";%
{\ar@<1.ex>@/^1.pc/|{\widehat Q}"a";"b"};%
{\ar@<1.ex>@/^1.pc/|{\widehat Q} "b";"a"};
\end{xy}
\end{equation*}
with $\widehat Q :=  L_q + Q$. Note that $\widehat Q$ is inhomogeneous in form degree, but is Grassmann odd and so we speak of a supercomplex. 
\begin{remark}
Naturally, the total supercomplex can be defined for general Q-algebroids.
\end{remark}
}
{
\subsection{Triangular Lie bialgebroids}
If $A$ is a Lie algebroid, then $\Pi A^*$ comes equipped with a weight $-1$ Schouten bracket, which we will denote as $\SN{-,-}$. A  \emph{Poisson structure} $\mathcal{P}$ on $A$ is a Grassmann even weight $2$ function on $\Pi A^*$ such that $\SN{\mathcal{P}, \mathcal{P}} =0$. Note that $Q_\mathcal{P} := \SN{\mathcal{P},-}$ is a weight one homological vector field on $\Pi A^*$ and so $A^*$ is a Lie algebroid. The graded Jacobi identity for the Schouten bracket means that $Q_\mathcal{P}\SN{X,Y} = \SN{Q_\mathcal{P} X, Y} + (-1)^{\widetilde{X}+1}\, \SN{X,Q_\mathcal{P} Y}$ for all $X$ and $Y\in C^\infty(\Pi A^*)$. Thus, $(A,A^*)$ is a \emph{triangular Lie bialgebroid} (here in the super-setting) as defined by Mackenzie and Xu (see \cite{Mackenzie:1994} and Kosmann-Schwarzbach \cite{KosmannSchwarzbach:1995} for details). \par
For any vector bundle $A$, we have the \emph{odd isomorphism map} (as $C^\infty(M)$-modules) which identifies sections of $A$ with weight one functions on $\Pi A^*$ 
$$\varsigma : \Sec(A) \mapsto \mathcal{A}^1(\Pi A^*)\,.$$
Using local coordinates $(x^a, \eta_\alpha)$ on $\Pi A^*$ the identification is given by $s = s^\alpha(x)t_\alpha \mapsto \varsigma(s) := (-1)^{\widetilde{s}}\, s^{\alpha}(x) \eta_\alpha$. Note $\widetilde{\varsigma(q)} = \widetilde{s} +1$.  Let $(A, q)$ be an inner Q-algebroid, and define $\mathcal{Q}:= \varsigma(q)$. The homological condition $[q,q] =0$ becomes $\SN{\mathcal{Q}, \mathcal{Q}}=0$. It is clear that $\SN{\mathcal{Q},-}$ is a homological vector field on $\Pi A^*$, but it is weight zero and so does not define a Lie algebroid structure.\par 
However, we can define $\mathcal{P}_q := \half \mathcal{Q}^2$, which, due to the Leibniz rule for the Schouten bracket, is easily seen to be a Poisson structure on $A$. Indeed, $\SN{\mathcal{P}_q, \mathcal{P}_q} \propto\SN{\mathcal{Q}, \mathcal{Q}}\, \mathcal{Q}^2 =0$. Thus, we have established the following.
\begin{theorem}\label{thm:BiAlg}
Let $(A,q)$ be an inner Q-algebroid. Then the pair $(A,A^*)$ is canonically a triangular Lie bialgebroid.
\end{theorem}
\begin{remark}
The reader should note the similarity with decomposable Poisson structures on purely even manifolds, i.e., Poisson structures of the from $\pi = X\wedge Y$, where $X$ and $Y$ are vector fields. It is clear that, in this classical situation, we cannot set $X=Y$ and obtain a non-vanishing Poisson structure. The situation is vastly different for supermanifolds. 
\end{remark}
\begin{example}
Let $(M,Q)$ be a Q-manifold and define the odd principal symbol $\mathcal{Q}:= \varsigma(Q) \in \mathcal{A}^1(\Pi \sT^* M)$. Then $\mathcal{P}_Q := \half \mathcal{Q}^2$ is a Poisson structure on $M$ and thus $(\sT M, \sT M^*)$ is a triangular Lie bialgebroid.
\end{example}
\begin{remark}
The construction of a Poisson structure from a homological section generalises directly to higher Poisson structures on Lie algebroids by taking any smooth function of $\mathcal{Q}$ (see \cite{Bruce:2010} for higher Poisson/Schouten structures on Lie algebroids). For example, $\mathcal{P}:= \exp(\mathcal{Q}) = 1 + \mathcal{Q} + \frac{1}{2!} \mathcal{Q}^2 + \frac{1}{3!} \mathcal{Q}^3 +  \cdots$ is higher Poisson structure on $A$, i.e., an even function on $\Pi A^*$ such that $\SN{\mathcal{P}, \mathcal{P}} =0$. 
\end{remark}
}
\subsection{Derived odd Loday--Leibniz brackets}
We now turn our attention { to} the induced or derived odd Loday--Leibniz structure on the space of sections of an { IQA}. { First, we give the definition of an odd Loday--Leibniz algebra following Kosmann-Schwarzbach \cite{KosmannSchwarzbach:1996}.

\begin{definition}\label{def:OddLodLeiAlg}
An \emph{odd Loday-Leibniz algebra} is a $\Z_2$-graded $\R$-vector space $\mathfrak{g} = \mathfrak{g}_{0}\oplus \mathfrak{g}_{1}$, equipped with a $\R$-bilinear operation (the \emph{odd Loday-Leibniz bracket})
$$(-,-) : \mathfrak{g} \times \mathfrak{g} \longrightarrow \mathfrak{g},$$
that satisfies the following.
\begin{enumerate}
\setlength\itemsep{1em}
\item $\widetilde{(u,v)} =  \widetilde{u} + \widetilde{v} + 1$, and 
\item $(u, (v,w))= ((u,v), w)+ (-1)^{(\widetilde{u}+ 1)( \widetilde{v} + 1)} (v, (u, w))$.
\end{enumerate}
For all $u,v$ and $w \in \mathfrak{g}$. 
\end{definition}
Note that we do not insist on the odd Loday-Leibniz having any form of (graded) skewsymmetry upon interchange of the elements.\par
}
The existence of such a bracket is almost a  direct corollary of the definition of an { IQA}. 
\begin{proposition}\label{trm:OddLodLei}
Let $(A,q)$ be an { IQA}. Then $\Sec(A)$ comes with an odd Loday--Leibniz bracket.
\end{proposition}
\begin{proof}
This follows from  Proposition \ref{prop:DiffLieAlg} and the derived bracket formalism (see \cite{KosmannSchwarzbach:1996}). 
\end{proof}
Specifically, we define 
$$(u,v) = (-1)^{\widetilde{u}} \, [[q,u],v],$$
with $u$ and $v \in \Sec(A)$. A quick calculation shows that 
$$(u, (v,w)) = ((u,v),w) + (-1)^{(\widetilde{u}+1)( \widetilde{v}+1) } \, (v,(u,w)).$$
As $(\Sec(A), [-,-])$ is, in general, non-Abelian, the resulting odd brackets are not (graded) skewsymmetric and so we have an odd Loday--Leibniz bracket. The violation or `anomaly' to the skewsymmetry is easily seen to be given by 
\begin{equation}\label{eqn:AnomSym}
(u,v) = - (-1)^{(\widetilde{u} +1)(\widetilde{v}+1)}\, (v,u) + (-1)^{\widetilde{u}} \delta [u,v].
\end{equation}
\begin{remark}
Loday--Leibniz algebras were introduced  with the name `Leibniz algebras' by Loday (see \cite{Loday:1993}) as a noncommutative analogue of Lie algebras, i.e., the bracket is no longer skewsymmetric.  In the current context, we have a Grassmann odd version of a Loday--Leibniz algebra, i.e., $\big(\Sec_i(A), \Sec_j(A)\big) \subset \Sec_{i+j+1}(A)$. We use the notation `$(-,-)$' for the odd Loday--Leibniz bracket as this is  commonly  used in the physics literature to denote the antibracket as found in the BV-BRST formalism. The reader should also keep in mind the Schouten--Nijenhuis bracket of multivector fields.  It is possible to consider a symmetrised version of an odd derived bracket, but this non-longer satisfies the Jacobi identity. For details the reader should consult the original literature \cite{KosmannSchwarzbach:1996}. 
\end{remark}
{ As we are dealing with sections of a vector bundle, we have right and left Leibniz rules for the odd Loday--Leibniz bracket. }
{
\begin{proposition}\label{prop:leftLeibRule}
The left Leibniz rule for the odd Loday--Leibniz bracket is
$$(u, f v) = \rho_L(u)f \, v + (-1)^{(\widetilde{u}+1)\widetilde{f}}\, f (u,v),$$
where $  \rho_L(u) := - [\rho(u), Q_q] = (-1)^{\widetilde{u}}\, L_{Q_q}\rho(u)$. Moreover, the left odd anchor $\rho_L : \Sec(A) \rightarrow \Vect(M)$ satisfies
$$\rho_L (u,v) = (-1)^{\widetilde{u} + \widetilde{v}} \, [\rho(\delta u), \rho(\delta v)],$$
for all $u$ and $v \in \Sec(A)$ and $f \in C^\infty(M)$.\par 
The right Leibniz rule for the odd Loday--Leibniz bracket is
$$ (fu, v) = f \, (u,v) - (-1)^{(\widetilde{u} + \widetilde{f} +1)(\widetilde{v} +1)}\, \rho_L(v)f \, u 
 + (-1)^{\widetilde{u}} \, \delta \, [fu,v]. $$
\end{proposition}
\begin{proof}
 Starting with the left Leibniz rule, a direct computations established the result.
\begin{align*}
(u, fv) &= (-1)^{\widetilde{u}} \, [\delta u , f v] = (-1)^{\widetilde{u}} \, \rho(\delta u)f \, v + (-1)^{\widetilde{u} + (\widetilde{u} +1)\widetilde{f}} \, f [\delta u, v]\\
& = (-1)^{\widetilde{u}}\, \rho\big( [q,u]\big)f \, v + (-1)^{ (\widetilde{u} +1)\widetilde{f}} \, f (u,v)\\
&= - [\rho(u), Q_q]f \, v + (-1)^{ (\widetilde{u} +1)\widetilde{f}} \, f (u,v).
\end{align*}
Similarly, a direct computation shows that
\begin{align*}
\rho_L(u,v) &= - [\rho(u,v), Q_q] = (-1)^{\widetilde{u}+1}\, [ [\rho(\delta u), \rho(v)], Q_q]\\
& = (-1)^{\widetilde{v} +1} \, [Q_q, [\rho(\delta u), \rho(v)]] = (-1)^{\widetilde{v} +1} \, [Q_q, [ [Q_q , \rho(u) ] , \rho(v) ]]\\
&=(-1)^{\widetilde{u} + \widetilde{v}} \, [[Q_q , \rho(u)] , [Q_p , \rho(v)]] = (-1)^{\widetilde{u} + \widetilde{v}} \, [\rho(\delta u), \rho(\delta v)].
\end{align*}
Moving to the right Leibniz rule, the proof follows directly using the `anomaly' to the symmetry of the odd Loday--Leibniz bracket \eqref{eqn:AnomSym}, i.e.,  the term $(-1)^{\widetilde{u}}\delta[u,v]$, and and the left Leibniz rule.  
\end{proof}}
As the `anomaly' to the symmetry of the odd Loday--Leibniz bracket is a $\delta$-exact term, it is unsurprising that the left and right Leibniz rules are similarly related. 
\begin{remark}
The odd Loday--Leibniz bracket structure on $\Sec(A)$ is a generalised version of a Loday  or Leibniz algebroid. There have been several non-equivalent definitions of a Loday algebroid in the literature, the reader can consult  \cite{Grabowski:2013,Jubin:2016} and references therein for details. The standard motivating example of a Loday algebroid is a Courant algebroid equipped with the Courant--Dorfman bracket (see \cite{Courant:1990,Dorfman:1987}). Often in literature one only requires that there be a left Leibniz rule in the definition of a Leibniz algebroid, see for example \cite{Ibanez:1999}. In our case, the lack of skewsymmetry is precisely controlled and so we can directly consider what happens with both the left and right Leibniz rules.
\end{remark}
\begin{definition}
Let $(A,q)$ be an { IQA}. A section $u \in \Sec(A)$ is said to be \emph{left central} with respect to the odd Loday--Leibniz bracket if $(u, v)=0$ for all $v \in \Sec(A)$.   Similarly, a section $u \in \Sec(A)$ is said to be \emph{right central} with respect to the odd Loday--Leibniz bracket if  $(v, u)=0$ for all $v \in \Sec(A)$. If a section is both left and right central with respect to the odd Loday--Leibniz bracket then it is said to be \emph{central} with respect to the odd Loday--Leibniz bracket.
\end{definition}
It is obvious that if a section $u \in \Sec(A)$ is central with respect to the Lie bracket, then it is central with respect to the odd Loday--Leibniz bracket. 
\begin{proposition}
Let $(A,q)$ be an { IQA}. The homological section $q$ is central with respect to the odd Loday--Leibniz bracket that it generates.
\end{proposition}
\begin{proof}
The left central property is obvious;  $(q,v) = - [[q,q],v]=0$ as $[q,q]=0$. The right central property follows from the skewsymmetry and Jacobi identity for the Lie bracket;
$$(v,q) = (-1)^{\widetilde{v}} \, [[q,v],q]  = [q,[q,v]]
 = [[q,q],v] - [q,[q,v] ] = - (-1)^{\widetilde{v}}[[q,v],q] = - (v,q).
$$
Thus, $(v,q)=0$ for any $v \in \Sec(A)$.
\end{proof}
\begin{example}
Let $(M,Q)$ be a Q-manifold, then we have the natural structure of an { IQA} on $\sT M$. The associated odd Loday--Leibniz bracket on $\Vect(M)$ is given by
$$(X,Y) = (-1)^{\widetilde{X}}\, [[Q,X],Y].$$
One can easily see that 
\begin{align*}
(X,fY) &= (-1)^{\widetilde{X}}\, L_Q X(f) \, Y + (-1)^{(\widetilde{X}+1)\widetilde{f}}\, f \, (X,Y),\\
(fX ,Y) & = f \, (X,Y) - (-1)^{(\widetilde{X} + \widetilde{f}+1)(\widetilde{Y}+1) + \widetilde{Y}}\, L_QY(f) \, X + (-1)^{\widetilde{X}}\, L_Q\big([fX, Y]\big).
\end{align*}
\end{example}
\begin{example}
As a more specific example of the previous example, consider a purely even Lie algebroid $A$. Then $(\Pi A, Q)$ is a Q-manifold and $(\sT \Pi A, Q)$ is an { IQA}. Lie algebroid (pseudo-)forms are functions on $\Pi A$, i.e., $\Omega^\bullet(A) = C^\infty(\Pi A)$, and in turn, $\Der(\Omega^\bullet(A)) \cong \Vect(\Pi A)$. Thus,  $\Der(\Omega^\bullet(A))$ comes canonically equipped with an odd Loday--Leibniz bracket. Restricting attention even further to the tangent bundle of a purely even manifold we see that the space of derivations of differential forms comes with  an odd Loday--Leibniz bracket generated by the de Rham differential. { If we restrict attention to weight $-1$ vector fields, then we are really dealing with the interior product of differential forms. In particular, the associated bracket, up to a shift in parity, is just the Lie bracket, $i_{ [X,Y]} =  [[\rmd, i_X],i_Y]$. The relation between the odd Loday--Leibniz bracket on differential forms and the Fr\"{o}licher--Nijenhuis bracket was made clear by Kosmann-Schwarzbach in \cite{KosmannSchwarzbach:2004} and we differ the reader there for details.}
\end{example}
{ Note that $\ker(\rho) \subset \Sec(A)$ is closed under the odd Loday--Leibniz. Moreover, the odd Loday--Leibniz restricted to $\ker(\rho)$ is graded skewsymmetric and the left and right anchors vanish, i.e., $(u,v) = -(-1)^{(\widetilde{u}+1)(\widetilde{v}+1)}\, (v,u)$ and $(u, f\, v) = (-1)^{(\widetilde{u}+1)\widetilde{f}} \, f \, (u,v)$, for all $u,v \in \Sec(A)$ and $f \in C^\infty(M)$. We can shift the parity of the sections\footnote{ i.e., $\Pi \Sec(A)_0 = \Sec(A)$ and $\Pi \Sec(A)_1 = \Sec(A)_0$ and the module structure is $f (\Pi u) = (-1)^{\widetilde{f}}\, \Pi (fu)$.} and then define a Lie bracket on $\Pi \ker(\rho)$ related to the odd Loday--Leibniz bracket via
$$\Pi \{x,y \} =  (\Pi x, \Pi y)\,$$
where $x$ and $y \in \Pi \ker(\rho)$. The graded symmetry is obvious. One can quickly check the Jacobi identity 
\begin{align*}
\Pi \{ x, \{y,z \}\} & = (\Pi x, (\Pi y, \Pi z)) = ((\Pi x ,\Pi y), \Pi z) + (-1)^{\widetilde{x} \widetilde{y}}\, (\Pi y , (\Pi x, \Pi z))\\
& \Pi \left( \{\{x,y\}z\} +(-1)^{\widetilde{x}\widetilde{y}}\, \{y , \{x,z \}\right)\,,
\end{align*}
for all $x,y$ and $z \in \Pi \ker(\rho)$. Moreover,  for any $f\in C^\infty(M)$ we have
$$\Pi \{x, fy \}  = (-1)^{\widetilde{f}}\,(\Pi x, f \Pi y) = (-1)^{\widetilde{f} (\widetilde{x}+1)}\, f \, (\Pi x, \Pi y) =  \Pi \left((-1)^{\widetilde{f}\widetilde{x}}\, f\, \{x,y \}\right)  \,.$$
Thus, we have established the following.
\begin{proposition}
Let $(A,q)$ be an inner Q-algebroid. Then $\mathfrak{g} :=  \big( \Pi \ker(\rho) , \{-,- \}\big)$ is a Lie algebra and $\big(C^\infty(M), \mathfrak{g}\big)$ is a Lie--Rinehart algebra with the action of $\mathfrak{g}$ on $C^\infty(M)$ is the trivial action, i.e., $\mathfrak{g} \rightarrow \Vect(M)$ is the zero map.
\end{proposition}
}
\subsection{Representations and differential modules}
Recall that an \emph{$A$-valued connection} on a vector bundle $\tau : E \rightarrow M$ is an even linear operator 
$$\nabla : \Sec(A) \times \Sec(E) \longrightarrow \Sec(E),$$
that satisfies the following;
\begin{enumerate}
\item $\nabla_{u+v}s = \nabla_u s + \nabla_v s$,
\item $\nabla_u(s + s') = \nabla_u s  + \nabla_u s'$,
\item $\nabla_{f u}s = f \nabla_u s$,
\item $\nabla_u fs = \rho(u)f \, s  + (-1)^{\widetilde{u} \widetilde{f}} f \nabla_u s$,
\end{enumerate}
for all $u$ and $v \in \Sec(A)$,  $s$ and $s' \in \Sec(E)$, and $f \in C^\infty(M)$. The \emph{curvature} of an $A$-value connection is 
$$R(u,v)s = [\nabla_u, \nabla_v]s - \nabla_{[u,v]}s.$$
\begin{remark}
 As far as we are aware, Lie algebroid connections were first introduced by Mackenzie \cite{Mackenzie:1987}. In the above definition we have, of course, generalised this slightly to the case of Lie superalgebroids. { Recall that a $A$-valued superconnection on $E$ (see Quillen \cite{Quillen:1985}) is an odd differential  operator $D$ on $\Omega^*(A)\otimes_{C^\infty(M)} \Sec(E)$ that satisfies $D(\omega \, s) = (Q\omega)\, s + (-1)^{\widetilde{\omega}} \, \omega \, (D s)$. Changing picture slightly, we identify $\Sec(E)$ with weight one functions on $E^*$. We can then consider a Quillen superconnection to be an odd vector field of bi-weight $(*,0)$ on the bi-graded supermanifold $\Pi A \times_M E^*$. In local coordinates $(x^a,\zx^\alpha, p_i)$, then
 $$D =  \zx^\alpha Q_\alpha^a(x)\frac{\partial}{\partial x^a} + \frac{1}{2}\zx^\alpha \zx^\beta Q_{\beta \alpha}^\gamma(x) \frac{\partial}{\partial \zx^\gamma} + \left(\mathbb{A}_i^{\,\,j}(x) + \zx^\alpha (\mathbb{A}_\alpha)_i^{\,\,j}(x) + \frac{1}{2}\zx^\alpha\zx^\beta (\mathbb{A}_{\beta\alpha})_i^{\,\,j}(x)+ \cdots \right )\, p_j \frac{\partial}{\partial p_i}\,.$$
 Restricting to Quillen superconnections of bi-weight $(1,0)$ we recover $A$-valued connections on $E$ as defined above.}
\end{remark}
A \emph{Lie algebroid representation} of $A$ on $E$ is a choice of flat $A$-valued connection, i.e., an $A$-valued connection such that $R(u,v)= 0$ for all $u$ and $v \in \Sec(A)$. Now suppose that $(A,q)$ is an { IQA} and that a Lie algebroid representation has been chosen.  Then we have a canonical odd endomorphism
$$\nabla_q : \Sec(E) \longrightarrow \Sec(E),$$
that satisfies 
\begin{align}\label{eqn:diffMod}
& \nabla_q fs = Q_q(f)\, s + (-1)^{\widetilde{f}}\, f \, \nabla_q s, && \nabla_q^2s = \half [\nabla_q,\nabla_q] s = \half \nabla_{[q,q]}s =0.
\end{align}
Thus, the sections of $E$ come with a differential that satisfies the natural generalisation of Leibniz rule. { The equations \eqref{eqn:diffMod} define the structure of a (super) differential  $C^\infty(M)$-module \footnote{For the definition of a differential (graded) module see \href{https://stacks.math.columbia.edu/tag/09JH}{Stacks-Project, Section 09JH}}.}\par 
 In other words, we have the following proposition.
\begin{proposition}\label{prop:EDiffMod}
Let $(A,q)$ be an { IQA}, $E$ be a vector bundle and let $\nabla$ be a Lie algebroid representation  of $A$ on $E$. Then $(\Sec(E), \, \nabla_q)$ is a differential $C^\infty(M)$-module. 
\end{proposition}
Thus, we have a supercomplex 
\begin{equation*}
\leavevmode
\begin{xy}(0,40)*+{\Sec_0(E)}="a"; (30,40)*+{\Sec_1(E)}="b";%
{\ar@<1.ex>@/^1.pc/|{\nabla_q}"a";"b"};%
{\ar@<1.ex>@/^1.pc/|{\nabla_q} "b";"a"};
\end{xy}
\end{equation*}
and an associated cohomology. Note that if we have a purely even Lie algebroid $A$, then $\nabla_q =0$ as the only homological section is the zero section. Similarly, if $E$ is purely even, then $\nabla_q$ is the zero map, i.e., sends all sections to the zero section. Thus, interesting cohomology groups can only exist if both $A$ and $E$ are both super. 
{
\begin{proposition}
Let $(A,q)$ be an IQA, $E$ be a vector bundle and let $\nabla$ be a Lie algebroid representation  of $A$ on $E$. Then 
$$\nabla_{\delta u}s =  [\nabla_q, \nabla_u]s\,$$
for all $u \in  \Sec(A)$ and $s \in \Sec(E)$.
\end{proposition}
\begin{proof}
This follows directly from $R(q,u)s=0$  for flat connections and $\delta u = [q,u]$. 
\end{proof}
\begin{corollary}
If $u \in \Sec(A)$ is $\delta$-closed, i.e., $\delta u =0$, then $[\nabla_q , \nabla_u] =0$.
\end{corollary}
}
\begin{example}
Let $(A,q)$ be an { IQA} over $M$ and assume that there exists a flat $A$-valued connection $\nabla$ on $\sT M$. Then $(\Vect(M), \nabla_q)$ is a  differential $C^\infty(M)$-module.
\end{example}
\begin{example}
Let $(E, \nabla)$ be a vector bundle over $M$ equipped with a flat linear connection.  Furthermore, if $(M,Q)$ is a Q-manifold, then $(\Sec(E) , \nabla_Q)$ is a differential $C^\infty(M)$-module. A little more specifically, if $E = \sT M$ and $\nabla$ is a flat affine connection then $(\Vect(\sT M) , \nabla_Q)$ is a  differential $C^\infty(M)$-module.
\end{example}
\begin{remark}
If the connection is not flat, then we have $\nabla_q^2 =\half R(q,q)$, where $R(q,q)$ is and even endomorphism of $\Sec(E)$. One can easily show that $\nabla_q^2 fs = f \,\nabla_q^2 s$. We then see that  $(\Sec(E), \, \nabla_q)$ is a  `curved differential $C^\infty(M)$-module'.
\end{remark}
If we consider an $A$-connection on $A$, for brevity we will say ``a connection on $A$'',  then we have the notion of the \emph{torsion} of a connection:
$$T(u,v) := \nabla_u v -(-1)^{\widetilde{u} \widetilde{v}} \, \nabla_v u -[u,v].$$
{
\begin{proposition}\label{prop:NabDel}
Let $(A,q)$ be an { IQA} equipped with a flat torsion-free connection $\nabla$, i.e., $T(u,v)=0$ for all $u$ and $v \in \Sec(A)$,  then
\begin{enumerate}
 \item $q$ is $\nabla_q$-closed,
 \item $\nabla_q = \delta$ provided $q$ is a parallel section, i.e., $\nabla_u q =0$ for all $u \in \Sec(A)$,\label{prop:NabDel2}
 \end{enumerate} 
\end{proposition}
\begin{proof}
Both of these follow directly from the definition of the torsion.
\end{proof}
\begin{remark}
Note that \eqref{prop:NabDel2} of Proposition \ref{prop:NabDel} does not require the connection to be flat.
\end{remark}
}
\begin{example}
Let $(M,Q)$ be a Q-manifold equipped with an affine connection (not necessarily flat).  Then $\nabla_Q Q =0$ if the affine connection is torsion-free. Furthermore, if the supermanifold is $n|2m$-dimensional or $n|n$-dimensional then it can be equipped with an even or odd Riemannian structure, respectively. In either case we  have the Levi-Civita connection, which just as in the classical setting is torsion-free {(see for example \cite{Monterde:1996})}.   
\end{example}

\subsection{Morphisms of  inner Q-algebroids}\label{subsec:Morph}
Let $(A, [-,-]_A, \rho_A)$ and $(B, [-,-]_B, \rho_B)$ be Lie algebroids. Recall that a vector bundle morphism 
 \begin{equation*}
\leavevmode
\begin{xy}
(0,20)*+{A}="a"; (20,20)*+{B}="b";%
(0,0)*+{M}="c"; (20,0)*+{N}="d";%
{\ar "a";"b"}?*!/_3mm/{\Phi};%
{\ar "c";"d"}?*!/^3mm/{\phi};%
{\ar "a";"c"}?*!/^3mm/{\pi_A};%
{\ar "b";"d"}?*!/_3mm/{\pi_B};%
\end{xy}
\end{equation*}
is a \emph{morphism of Lie algebroids} (see \cite{Higgins:1990}) if and only if $\Phi^\Pi : \Pi A \rightarrow \Pi B$ is a morphism of Q-manifolds (see \cite{Vaintrob:1996}), i.e., 
\begin{equation}\label{eqn:LieAlgMor}
Q_A \circ (\Phi^\Pi)^* = (\Phi^\Pi)^* \circ Q_B\,.
\end{equation}
In order to establish a generalisation of \eqref{eqn:LieAlgMor} for odd sections of a vector bundle we need to change picture slightly and consider $s\in \Sec_1(A)$ as an algebra morphism
$$(s^\Pi)^* : \cO_{\Pi A}(|\Pi A|) \longrightarrow \cO_M(|M|).$$
Note that the shift in parity is essential for odd sections. Locally we can write
$$(s^\Pi)^*(x^a, \zx^\alpha) = (x^a , s^\alpha(x)).$$
\begin{definition}\label{def:PhiRel}
Let $s \in \Sec_1(A)$ and $r \in \Sec_1(B)$ be two odd sections of the vector bundles $A$ and $B$, respectively. Furthermore, let  $\Phi: A \rightarrow B$ be a morphisms of vector bundles. Then $s$ and $r$ are said to be \emph{$\Phi$-related} if and only if
$$(s^\Pi)^* \circ (\Phi^\Pi)^* = \phi^* \circ(r^\Pi)^*\,.$$ 
\end{definition}
Locally we see that 
\begin{align*}
& (s^\Pi)^*\left((\Phi^\Pi)^*(y^i, \theta^\mu)\right) = \big(\phi^i(x), s^\alpha(x) \Phi_\alpha^{\,\, \mu}(x) \big),\\
& \phi^*\left( (r^\Pi)^* (y^i, \theta^\mu) \right) = \big(\phi^i(x), r^\mu(\phi(x)) \big),
\end{align*}
thus we require 
\begin{equation}\label{eqn:PhiRel}
s^\alpha(x)\Phi_\alpha^{\,\, \mu}(x) =r^\mu(\phi(x)).
\end{equation}
\begin{definition}\label{def:QLieMor}
Let $(A,  q_A)$ and $(B, q_B)$ be { IQAs}. Then a vector bundle morphism $\Phi : A \rightarrow B$ is a \emph{morphism of { inner Q-algebroids}} if and only if
\begin{enumerate}
\item $\Phi$ is a Lie algebroid morphism, and
\item $q_A$ and $q_B$ are $\Phi$-related.
\end{enumerate}
\end{definition}
Evidently, we can compose morphisms of { IQAs} and so we can define the category of { inner Q-algebroids} in the obvious way. The notion of an isomoprhism of { IQAs} is clear.\par 
We remind the reader that a morphism of Lie algebroids intertwines the respective anchor maps and that in local coordinates this condition is  expressed as 
\begin{equation}\label{eqn:AncMorph}
\rho_\alpha^a(x) \frac{\partial \phi^i(x)}{\partial x^a} = \Phi_\alpha^{\,\, \mu}(x) \rho_\mu^i(\phi(x)).
\end{equation}
As before we define $Q_M := \rho_A(q_A) = q_A^\alpha\rho^a_\alpha\frac{\partial}{\partial x^a}$  and $Q_N := \rho_B(q_B) = q_B^\mu\rho^i_\mu\frac{\partial}{\partial y^i}$.
\begin{proposition}\label{prop:Qmorph}
Let $(A, q_A)$ and $(B,q_B)$ be { IQAs}. Then if $\Phi: A \rightarrow B$ is a morphism of { inner Q-algebroids} (see Definition \ref{def:QLieMor}) then 
$$\phi :  (M,Q_M) \longrightarrow (N, Q_N),$$
is a morphisms of Q-manifolds, i.e.,
$$Q_M \circ \phi^* = \phi^* \circ Q_N\,. $$
\end{proposition} 
\begin{proof}
We will work locally and use \eqref{eqn:AncMorph} and \eqref{eqn:PhiRel}. Let $f \in C^\infty(N)$ be an arbitrary function. Then
\begin{align*}
Q_M(\phi^*f) & = q_A^\alpha \rho_\alpha^a \frac{\partial \phi^i}{\partial x^a} \, \phi^* \left( \frac{\partial f}{\partial y^i}\right)
= q_A^\alpha \Phi_\alpha^{\,\, \mu} \,  \phi^*\left(\rho_\mu^i \frac{\partial f}{\partial y^i} \right)\\
&= \phi^*\left(q_B^\mu \rho_\mu^i \frac{\partial f}{\partial y^i} \right) = \phi^* \big(Q_N(f) \big).
\end{align*}
\end{proof}
{
\begin{proposition}
Let $(A, q_A)$ and $(B,q_B)$ be  IQAs. Then if $\Phi: A \rightarrow B$ is a morphism of  inner Q-algebroids (see Definition \ref{def:QLieMor}) then
\begin{enumerate}
\item \label{prop:irelated}$\iota_{q_A} \circ (\Phi^\Pi)^* =  (\Phi^\Pi)^* \circ \iota_{q_B}$, and
\item \label{prop:Lrelated}$L_{q_A} \circ (\Phi^\Pi)^* =  (\Phi^\Pi)^* \circ L_{q_B}$.
\end{enumerate}
\end{proposition}
\begin{proof}
Directly in local coordinates we observe that
\begin{align*}
- q^\alpha_A(x) \frac{\partial}{\partial \zx^\alpha}\, (\Phi^\Pi)^* f(y, \theta)& = - q^\alpha_A (x) \Phi^\mu_\alpha (x) \, (\Phi^\Pi)^*\left( \frac{\partial f(y, \theta)}{\partial \theta^\mu}\right)\\
& = - (\Phi^\Pi)^* \left(q_B^\mu(y)\frac{\partial f(y, \theta)}{\partial \theta^\mu} \right) \,,
\end{align*}
thus \eqref{prop:irelated} is established. Part \eqref{prop:Lrelated} follows from the properties of $\Phi$-related vector fields and the Lie bracket. Specifically,   $[Q_A, \iota_{q_A}] \circ(\Phi^\Pi)^* = (\Phi^\Pi)^*\circ [Q_B, \iota_{q_B}]$ and so \eqref{prop:Lrelated} is established.
\end{proof}

Then we observe that morphisms of inner Q-algebroids give rise to morphisms of double Q-manifolds, i.e., the homological vector fields  of different weights defining a Q-algebroid are $\Phi$-related. 
}
For base preserving morphisms $\Phi: A \rightarrow B$, we have a $C^\infty(M)$-module homomorphism (using the same symbol) $\Phi : \Sec(A) \rightarrow \Sec(B)$ such that 
\begin{enumerate}
\item $\rho_A = \rho_B \circ \Phi$,
\item $\Phi [u,v]_A = [\Phi(u), \Phi(v)]_B$, and
\item $\Phi(q_A) = q_B$.  
\end{enumerate}
Naturally, we have $\rho_A(q_A) =\rho_B(q_B) = Q_M$.  { A quick calculation show that for base preserving morphisms of IQAs $\Phi \circ \delta_A =  \delta_B \circ \Phi$. }
\begin{proposition}
Let $(A, q_A)$ and $(B,q_B)$ be { IQAs} over the same base $M$. Suppose that $\Phi: A \rightarrow B$ is a morphism of { IQAs} (acting as the identity on $M$). Then 
\begin{enumerate}
\item $\Phi(u,v)_A = \big(\Phi(u), \Phi(v) \big)_B$, and
\item $\rho_{A,L}(u) = \rho_{B,L}\big(\Phi(u)\big)$, 
\end{enumerate}
for all $u$ and $v \in \Sec(A)$.
\end{proposition}
\begin{proof}
These follow as direct consequences of the definition of odd Loday--Leibniz brackets  as derived brackets.  Specifically,
\begin{enumerate}
\item $\Phi( u,v )_A = (-1)^{\widetilde{u}} \, \Phi \big( [[q_A, u]_A, v]_A\big) =  (-1)^{\widetilde{u}} \,  \big( [[\Phi(q_A), \Phi(u)]_B, \Phi(v)]_B\big)  = \big(\Phi(u), \Phi(v)  \big)_B$.
\item Considering $\Phi( u,fv )_A $ together with Proposition \ref{prop:leftLeibRule} and the above shows that 
$$\rho_{A,L}(u)f \, \Phi(v) \pm f \, \Phi(u,v)_A = \rho_{B,L}\big(\Phi(u)\big)f \pm f \, (\Phi(u), \Phi(v))_B,$$
\end{enumerate}
and so the desired result is established.
\end{proof}
\begin{remark}
Morphisms of Lie algebroids over different bases are not easy to describe in classical terms of brackets and anchors (see \cite[Section 4.3]{Mackenzie:2005}). The source of the problem is that one cannot pushforward sections and so compatibility of the brackets is a non-obvious statement. Due to this we will refrain from spelling out morphisms of {IQAs} over different bases in these terms.   
\end{remark}

\subsection{Modular classes of  inner Q-algebroids}
Recall that the \emph{modular class} of a Q-manifold $(M,Q)$ is the standard cohomology class of the divergence of the homological vector field \cite{Lyakhovich:2010,Lyakhovich:2004}, i.e.,
$$\textnormal{Mod}(Q) := [\textnormal{Div}_{\p} Q]_\textnormal{st}.$$
The modular class is independent of any chosen Berezin volume $\p$ as any other choice of Berezin volume leads to a divergence that differs by something $Q$-exact. The vanishing of the modular class is the necessary and sufficient condition for the existence of a $Q$-invariant Berezin volume. Q-manifolds with vanishing modular class are known as \emph{unimodular Q-manifolds}. The \emph{local characteristic representative} of the modular class is the local function 
\begin{equation}\label{eqn:LocRep}
\phi_Q(x) = \frac{\partial Q^a}{\partial x^a}.
\end{equation}
Note that this is not invariant under changes of local coordinates, only the divergence of $Q$ is truly invariant. However, as we are always dropping $Q$-exact terms, the local characteristic representative is meaningful. For a review of modular classes of Q-manifolds and further references the reader may consult \cite{Bruce:2017}.\par 
The modular class of a Lie algebroid (see \cite{Evens:1999}) can be understood/defined in precisely these terms - the earliest  reference that we are aware of where this is spelled out is Grabowski \cite{Grabowski:2012}. An { IQA}  $(A, q)$ can be considered as a Q-manifold in two different, but related, ways, either as $(\Pi A, Q)$ or $(\Pi A, L_q)$. We wish to now briefly examine the relation between the respective modular classes. \par 
{
\begin{remark}
Naturally, we have the notion of the modular classes of a general Q-algebroid.  Here we will only consider inner Q-algebroids.
\end{remark}
}
Using \eqref{eq:Q} and \eqref{eq:Lq} the local characteristic representatives are 
\begin{align*}
& \phi_Q(x, \zx) = \zx^\alpha \left((-1)^{\widetilde{a}(\widetilde{\alpha}+1)} \, \frac{\partial Q_\alpha^a}{\partial x^a}(x) + Q^\beta_{\alpha \beta}(x) \right),
&& \phi_q(x) = q^\alpha (x)\left((-1)^{\widetilde{a}(\widetilde{\alpha}+1)} \, \frac{\partial Q_\alpha^a}{\partial x^a}(x) + Q^\beta_{\alpha \beta}(x) \right).
\end{align*}
The above can then be written as
\begin{equation}\label{eq:LocCharRepRel}
\phi_q(x) = - \iota_q \phi_Q(x, \zx).
\end{equation}
\begin{proposition}\label{prop:ModClass}
Let $(A, q)$ be an { IQA}. Then if the modular class of $(\Pi A, Q)$ vanishes, then so does the modular class of $(\Pi A, L_q)$.
\end{proposition}
 \begin{proof}
This is obvious in light of \eqref{eq:LocCharRepRel}.
 \end{proof}
 \begin{corollary}\label{cor:InvBerVol}
 If the modular class of $(\Pi A, Q)$ vanishes, then there exists a Berezin volume on $\Pi A$ that is both $Q$-invariant and $L_q$-invariant.
 \end{corollary}
 { Recall that $Q_q = \rho(q) =  q^\alpha Q_\alpha^a \frac{\partial}{\partial x^a}$, and thus 
 $$\phi_{Q_q}(x) =  \phi_q(x) - q^\alpha(x) Q_{\alpha \beta}^\beta(x) +\left( \frac{\partial q^\alpha(x)}{\partial x^a}\right) Q_\alpha^a(x)\,.$$
 }
 \begin{proposition}
 Let $(A,q)$ be an { IQA}. Then 
 \begin{enumerate}\label{eqn:Divs}
 \item $Q\big( \textnormal{Div}_{\p}L_q\big) +L_q\big( \textnormal{Div}_{\p}Q\big)=0$, and
 \item $Q(\phi_q) + L_q (\phi_Q) =0$, where we have not written the obvious restriction. 
 \end{enumerate}
 \end{proposition}
 \begin{proof}\
 \begin{enumerate}
 \item From Remark \ref{rem:LinftyALg}, $\widehat{Q} = L_q + Q$ is a homological vector field  and so $\textnormal{Div}_{\p} \widehat{Q}$ is $\widehat{Q}$-closed. Using the linear properties of the divergence we see that
 $$\widehat{Q} \big(\textnormal{Div}_{\p} \widehat{Q}\big)= L_q \big(\textnormal{Div}_{\p} L_q\big) +  L_q \big(\textnormal{Div}_{\p} Q\big) +  Q \big(\textnormal{Div}_{\p} L_q\big) +  Q \big(\textnormal{Div}_{\p} Q\big) =0.$$
 Then as $L_q \big(\textnormal{Div}_{\p} L_q\big) =   Q \big(\textnormal{Div}_{\p} Q\big) =0$, obtain the desired result.
 \item This follows from (1) using the local description of the divergence. We choose some Berezin volume $\p = D[x]\, \rho(x)$, where $D[x]$ is the coordinate Berezin density.  Explicitly 
$$Q\big(\phi_q  + L_q \log(\rho)\big) + L_q\big(\phi_Q  +Q \log(\rho)\big) = Q(\phi_q) + L_q(\phi_Q) + [Q,L_q](\log(\rho)) = 0.$$
Using the fact that $[Q,L_q]=0$ establishes the result.
\end{enumerate}
 \end{proof}
 {
Taking the relevant part of the double complex \eqref{eqn:bicomp}
$$ \Omega^0_1(A) \stackrel{Q}{\longrightarrow} \Omega^1_0(A) \stackrel{L_q}{\longleftarrow} \Omega^1_1(A)\,,$$
we observe that 
$$\textnormal{Div}_{\p} L_q \longmapsto Q\big(\textnormal{Div}_{\p} L_q \big) = - L_q\big(\textnormal{Div}_{\rho} Q \big) \longmapsfrom \textnormal{Div}_{\p} Q\;.$$
 }
 \subsection{Comments on almost Lie algebroids}
If one drops the Jacobi identity for the bracket on sections of an anchored vector bundle, but keep the skewsymmetry, Leibniz rule, and the compatibility of the anchor and bracket we have a so-called \emph{almost Lie algebroid} (see \cite{Grabowski:2011,Grabowski:1999}). We remark that almost Lie algebroids naturally appear in geometric mechanics. The notion of a homological section of an almost Lie algebroid is exactly the same as a homological section of a Lie algebroid, we can speak of an inner almost Q-algebroid (IAQA).  Due to the compatibility of the bracket and anchor, the base manifold is again a Q-manifold (see Proposition \ref{prop:Qman}). However, $[q,-]$ is no longer a differential.  Thus, the derived bracket formalism does not work here. We cannot construct a derived bracket on the space of sections of a IAQA as we did for an { IQA}.  In terms of Vaintrob's description, for almost Lie algebroids the homological condition is weakened to $Q^2f =0$, for all functions $f$ on the base manifold (see \cite[Proposition 2.9]{Bruce:2020}). Thus, one can speak of exact one-forms and the modular class of an almost Lie algebroid is well-defined.  In general, we have a quasi-bi-complex associated with a IAQA, rather than a genuine bi-complex.  Other than the already highlighted points, the rest of this paper generalises directly to IAQAs.\par
 If we further drop the compatibility of the bracket and anchor, we have a \emph{skew algebroid}. While the notion of a homological section here still makes sense, there is little structure to work with. For example, the base manifold need not be a Q-manifold. It seems that while the Jacobi identity can be relaxed and interesting ``homological structures'' can still be defined, relaxing the compatibility of the bracket and anchor is too harsh. 

\section{Closing remarks}
In this paper, we gave the notion of a homological section of a Lie (super)algebroid and examined some direct consequences thereof.  We again stress that rephrasing  algebro-geometric structures in terms of Q-manifolds has been very fruitful and is expected to continue to be so. Speculatively, we suggest that novel and interesting structures could be encoded in { inner Q-algebroids}, possibly equipped with  a further structure such as a compatible $\mathbb{N}$-grading. { Other generalisations naturally present themselves, such as multiple compatible (linearly independent) homological sections of a Lie algebroid and extending these concepts to double Lie algebroids.} This all awaits to be explored.  
\section*{Acknowledgements}
\noindent The author thanks Janusz Grabowski, Rajan Mehta, Damjan Pi\v{s}talo and Norbert Poncin for their valuable comments. The author especially thanks the two anonymous referees for their comments that have served to greatly improve this paper.

\end{document}